\documentclass[copyright,creativecommons]{eptcs}
\providecommand{\event}{AFL 2014} 
\usepackage{breakurl}             

\usepackage{amssymb}
\usepackage{amsmath}
\usepackage{amsthm}
\usepackage{mathrsfs}
\usepackage{stmaryrd}
\usepackage{times}
\usepackage{latexsym}
\usepackage{hhline}
\usepackage{enumerate} 
\usepackage[latin1]{inputenc} 
\usepackage{calc}         
\usepackage{graphicx}     
\usepackage{ifthen}       
\usepackage{pst-all}      
\usepackage{pst-poly}     
\usepackage{multido}      

\usepackage[paper]{gastex}

\newcommand{\scaffolding}[1]{\ensuremath{\textit{sc-}#1}}
\newcommand{\filling}[1]{\ensuremath{\textit{fl-}#1}}

\newcommand{\C} {{\ensuremath{\textrm{COM-SLIP}^{\cup,\cdot}\;}}}

\newcommand{\shuffle}{\,\rotatebox[origin=c]{-90}{$\exists$}\, }

\theoremstyle{plain}

\newtheorem{proposition}{Proposition}
\newtheorem{theorem}{Theorem}

\newtheorem{lemma}{Lemma}
\theoremstyle{definition}

\newtheorem{definition}{Definition}
\newtheorem{example}{Example}
\def\titlerunning{Commutative  Languages and their  Composition by Consensual Methods}
\def\authorrunning{S. {Crespi Reghizzi} and P. {San Pietro}}

\title{Commutative  Languages and their  Composition by Consensual Methods\footnote{Work partially supported by \emph{PRIN 2010LYA9RH-006} ``Automi e linguaggi formali: Aspetti Matematici e Applicativi''.}\footnote{The main results have been announced in~\cite{conf/ictcs/Crespi-Reghizzi13}, with preliminary  sketchy  proofs entirely superseded by the present ones.}}
\author{Stefano {Crespi Reghizzi}
\and 
Pierluigi {San Pietro}
\institute{DEIB, Politecnico di Milano and CNR-IEIIT}
\email{stefano.crespireghizzi@polimi.it \quad pierluigi.sanpietro@polimi.it}
}

\begin{document}
\maketitle              

\begin{abstract}
Commutative  languages with the semilinear property (SLIP)  can be naturally recognized by real-time NLOG-SPACE multi-counter machines. We show that unions and concatenations of such languages can be  similarly recognized, relying on -- and further developing, our recent results on the family of consensually regular (CREG) languages.  A CREG language is defined by a regular language on the alphabet that includes the terminal alphabet and its marked copy. New  conditions, for ensuring that the union or concatenation of CREG languages is  closed,  are presented and applied to the commutative SLIP languages. The paper contributes to the knowledge of the CREG family, and introduces novel   techniques for language composition, based on arithmetic congruences that act as language signatures. Open problems are listed.

\end{abstract}

\section{Introduction}\label{sectIntroduction}
This paper focuses on commutative languages having the semilinear property (SLIP). We recall that a  language has the \emph{linear property} (LIP)  if, in any word,  the number of letter occurrences (also named Parikh image)  satisfies a linear equation; 
it has  the \emph{semilinear property} (SLIP)~\cite{Ginsburgh1966} if the number satisfies one out of finitely many linear equations. 
A language is \emph{commutative} (COM) if, for every word, all permutations
are in the language; 
thus, the legality of a word is based only on the Parikh image, not on the  positions of the letters.  
Here we deal with the subclass of COM languages enjoying the SLIP, denoted by COM-SLIP, for which we recall some known properties.
For a  binary alphabet, COM-SLIP languages are context-free 
whereas, in the general  case, they can be recognized by \emph{multi-counter machines} (MCM), in particular by non-deterministic quasi-real-time \emph{blind} MCM (equivalent to \emph{reversal-bounded} MCM~\cite{DBLP:journals/tcs/Greibach78a}). 
The COM-SLIP family is closed under all Boolean operations, homomorphism and inverse  homomorphism,  
but it is  not closed under concatenation. 
\par
Our  contribution  is to relate two seemingly disparate language families:  on one hand, the  COM-SLIP languages and their closure under union and concatenation (denoted by \C), on the other hand, the family of \emph{consensually regular} languages (CREG), recently introduced by the authors, to be later presented. 
We  briefly explain the intuition behind it. 
 Given a terminal alphabet, a CREG language  is specified by means of a regular language (the \emph{base})  
having a \emph{double} alphabet: the original one and a \emph{dotted} copy. Two or more words in the base language  \emph{match}, if they 
are all identical when the dots are disregarded and, in every position, exactly one word 
 has an undotted letter  (thus in all remaining words the same position is dotted). 
In our metaphor, we say that, position by position, one of the base words ``places''  a letter and the remaining words ``consent'' to it. 
A word is in the consensual language if the base language contains a set of matching words, identical to the given word when the dots are disregarded. 
This mechanism somewhat resembles the model of alternating non-deterministic finite automata, 
 but the criterion by which the parallel computations match is  more flexible and produces a recognition device which is a MCM working in NLOG-SPACE. 
This MCM can be viewed as a token or multi-set machine; it has one counter for each state of the DFA recognizing the base language;  each counter value  counts the number of  parallel threads that are currently active in each state.
 Our main result is that the \C family  is strictly included in CREG;
we also prove some non-closure properties of  \C.  
\par
To construct the regular language that serves as base  for the consensual definition of a \C \;  language,  we have devised a new method,  which may be also useful to study the inclusion in consensual classes of  other families closed union or concatenation.  
It is easy to consensually specify a COM-LIP language by means of a regular base; 
however, in general, union or concatenation of two regular bases  consensually specifies a larger language than the union or concatenation of the components.  
To prevent this to happen, we assign a distinct numeric congruence class to each base, which determines the positions where a letter may be placed  as dotted or as undotted. For a given word, such positions are not the letter orders, but they are the orders of the letters in the projections of the word on each letter of the alphabet. The congruence acts as a sort of signature that cannot be mismatched with other signatures. 
\par
To hint to a potential  application,  \C offers a
rather suitable  schema for certain  parallel computation systems, such as Valiant's
``bulk synchronous parallel computer''~\cite{ValiantLes1990a}. 
There, when all threads in a parallel computational phase, which we suggest to
model by a commutative language, terminate, the next phase can start; the
sequential composition of such phases can be represented by language concatenation;
and the composition of alternative subsystems can be modeled by  language  union. As said, such computation schema is not finite-state but it is a MCM.
\par
Paper organization: Sect.~\ref{SectionFirstDef} contains preliminaries, some simple properties of \C  and the consensual model. 
Sect.~\ref{SectDecomposedBases} introduces the  decomposed form, states and proves the conditions that ensure  union- and concatenation-closure, and details the  congruence based constructions. 
Sect.~\ref{SectDisciplined} proves the main result through  a series of lemmas. 
 The last section refers to related work and mentions some unanswered  questions. 
 
\section{Preliminary Definitions and Properties}\label{SectionFirstDef}\label{SectionPreliminaries}
 The terminal alphabet is denoted by $\Sigma = \{a_1, \ldots, a_k \}$,   the  empty word by $\epsilon$  and $|x|$ is the length of a word $x$.
The projection of $x$ on  $\Delta \subseteq \Sigma$ is denoted by
$\pi_\Delta\left( x \right)$; $|x|_a$ is shorthand for $|\pi_{\{a\}}\left( x \right)|$ for $a\in\Sigma$, and $|x|_\Delta$ stands for $|\pi_\Delta\left( x \right)|$.
 The $i$-th letter of $x$ is  $x(i)$ and $x(i,j)$ is the substring $x(i)\ldots x(j)$, $1\leq i \leq j \leq |x|$.  
 The \emph{shuffle} operation is denoted  by $\shuffle$. 
\par\noindent
The \emph{Parikh} \emph{image} or \emph{vector} of a word $x\in \Sigma^*$  is 
$\Psi(x)=\left[|x|_{a_1},\,\ldots,\,|x|_{a_k}\right]$; it can be
naturally  extended to a language. The component-wise addition of two vectors is denoted by $\vec{p'} + \vec{p''}$.
The \emph{commutative closure} of  $L\in \Sigma^*$ is  $com(L) = \{ x\in \Sigma^* \mid \Psi(x) \in \Psi(L) \}$.
 A language $L$  is \emph{commutative} if $com(L)=L$;
the
corresponding language family is named COM. 
A language $L\subseteq \Sigma^*$  has the \emph{linear property} (LIP) if there exist
$q+1>0$ vectors $\vec{c}, \vec{p}^{(1)}, \ldots,
\vec{p}^{(q)}$ over $\mathbb{N}^k$, (resp. the \emph{constant} and the \emph{periods})  
such
that  $\Psi(L) = \left\{ \vec{c} +n_1 \cdot\vec{p}^{(1)}+\ldots +  n_q \cdot
\vec{p}^{(q)} \mid n_1, \ldots, n_q \geq 0 \right\}$.
\\
A language has the \emph{semilinear property} (SLIP) if it is the finite  union
of  LIP languages.
The families of commutative LIP/SLIP languages are denoted by
\emph{COM-LIP}/ \emph{COM-SLIP}, respectively. 
It is  well known that COM-SLIP is
closed under the Boolean operations, inverse  homomorphism, homomorphism and Kleene star,  but  not under  concatenation, which in general destroys commutativity.
However, the concatenation of  COM-SLIP languages
still enjoys the SLIP. 
\par\noindent
Let $\C$ be the smallest family including COM-SLIP languages and closed under union and concatenation. 
 Let BLIND denote the class of languages accepted by nondeterministic, blind 
multicounter machines~\cite{DBLP:journals/tcs/Greibach78a}, which, we recall, are restricted to perform a test for zero only at the end of a computation; they are equivalent to reversal-bounded counter machines.
The following facts, although to our knowledge not stated in the literature, are straightforward.
\begin{proposition}\label{proposPropertyCATCOMSLIP} {\em Main Properties of \C.}
\begin{enumerate}
	\item Every \C language on a binary alphabet is context-free.
	\item \C $\subsetneq$ BLIND.
	\item The \C family is not closed  under  intersection 
and Kleene star.
\end{enumerate}
\end{proposition}
\begin{proof}
Let $L' = com\left((ab)^+\right)$.
Statement~(1) is immediate: since all COM-SLIP on a a binary alphabet are context-free~\cite{Latteux79,Rigo:2003}, also their union and concatenation is context-free. 
Statement~(2) is also immediate, since COM-SLIP is clearly included in BLIND, and 
BLIND is closed by union and concatenation. The inclusion is strict since BLIND includes also 
non-context-free languages on a binary alphabet~\cite{DBLP:journals/tcs/Greibach78a}. 
To prove non-closure of intersection -- Statement~(3) -- assume by contradiction that  the language $L_0 = L'\cap a^+ b^+  =\{a^n b^n\mid n>0\}$ is in \C. 
Hence, also the languages 
$L_1=\{a^+ b^n a^n \mid n>0\}$, $L_2 =\{a^m b^m a^+ \mid  m>0\}$ and $L_1 \cap L_2 = \{a^n b^n a^n \mid  n>0\}$ are in \C. 
But the latter language is  not context-free, contradicting Statement~(1).  
To complete the proof of Statement~(3), if \C were closed under Kleene star, then language $L_3 =(L' c)^*$ would be \C, with $c \not\in\{a,b\}$. 
However, \C is included in BLIND, which is an intersection-closed full semiAFL (see Section~5 of~\cite{Baker1974} and also Theorem~1 of~\cite{DBLP:journals/tcs/Greibach78a}), 
i.e., BLIND is closed under intersection, union, arbitrary homomorphism, 
inverse homomorphism, and intersection with regular languages. 
Hence, the language $L_4= L_3\cap (a^+b^+c)^* = \{a^n b^n c \mid n>0 \}^*$ would be in BLIND. 
Letter $c$ can be deleted by a homomorphism, hence also
the language $\{a^n b^n \mid n>0 \}^*$, is BLIND, contradicting Corollary~3 of~\cite{Baker1974} and also Theorem~6, Part (2), of~\cite{DBLP:journals/tcs/Greibach78a}.  
\end{proof}

\subsection{Consensual Languages.}\label{sec-consensual}
We present the necessary elements of consensual language theory 
\cite{journals/ita/Crespi-ReghizziP11,CrespiSanPietroCIAA2012}. 
Let $\mathring{\Sigma}$ be the {\em dotted} (or marked) copy of alphabet
  $\Sigma$.  For
each $a\in \Sigma$,   $\tilde{a}$ denotes the set $\{a, \mathring{a}\}$. 
The  alphabet $\widetilde{\Sigma}=\Sigma\cup\mathring{\Sigma}$ is named
 {\em double} (or internal).
To express a sort of agreement  between words over the double
alphabet, we  introduce a binary relation, called \emph{match}, over
$\widetilde{\Sigma}^*$. 
\begin{definition}[Match]\label{matchOperation}
The partial, symmetrical, and associative binary operator, called \textit{match},
$@:\widetilde\Sigma \times \widetilde\Sigma \rightarrow \widetilde\Sigma
$
is defined as follows, for all $a\in\Sigma$:
\[
\left\{%
\begin{array}{ll}
    a @ \mathring{a}=\mathring{a}@ a = a  & \\
    \mathring{a}@ \mathring{a}= \mathring{a} & \\
     \text{undefined} & \hbox{\text{in every other case}.} \\
\end{array}%
\right.    
\]
The  match is naturally extended to strings of equal length, as a letter-by-letter application, by assuming $\epsilon @ \epsilon = \epsilon$: for every $n>1$, for all $w, w'\in
\widetilde\Sigma^n$, if $w(i)@w'(i)$ is defined for every $i, 1\le i \le n$, then 
\[
w\; @\; w' = \left(w(1) @w'(1)\right) \cdot \ldots \cdot \left(w(n) @w'(n)\right).\quad \text{In every other case, $w@w'$ is undefined.}
\]
\end{definition}
\noindent Hence, the match is undefined on strings $w,w'$ of unequal lengths,  or else if there exists a position $j$ such that  $w(j) @ w'(j)$ is  undefined, which  occurs in three cases: when both
 characters are in $\Sigma$, when both are in $\mathring{\Sigma}$  and differ, and when either one is dotted but is not the dotted copy of the  other. 
Syntactically, the precedence of the  match operator is just under the precedence of the concatenation. 
 The match $w$ of two or more strings is further qualified as  {\em strong}  if $w\in\Sigma^*$, or
as 
{\em weak} otherwise. 
By Def.~\ref{matchOperation}, 
if $w=w_1 @ w_2 @ \ldots @ w_m$  is a strong match of $m\ge 1$ words $w_1, \dots, w_m$, then in each position $1\leq i\leq |w|$, exactly one
word,
say $w_h$, is undotted, i.e., $w_h(i)\in \Sigma$, and $w_j(i)\in
\mathring{\Sigma}$ for all $j\neq h$; we say that word $w_h$ {\em places}
the letter at position $i$ and the other words {\em consent} to it.
Metaphorically, the words that strongly match 
provide mutual consensus on the validity of the corresponding  word over
$\Sigma$, thereby motivating the name ``consensual'' of the language family. 
\\
The match  is extended to 
two languages $B', B''$ on the double alphabet,  as $ 
B'\, @\, B'' = \{w'\, @\, w''\mid w' \in B', w'' \in B''\}
$. The iterated match $B^{i@}$ is defined for all $i\ge 0$, as $B^{0@}=B$, $B^{i@}=B^{(i-1)@}@B$, if $i>0$. 
\begin{definition}[Consensual
language]\label{closureUnderMatch}\label{defConsensualLang}
The \emph{closure under match}, or $@$-\emph{closure}, of a language $B\subseteq
\widetilde\Sigma^*$ is \; 
$
B^@ = \bigcup_{i\ge0} B^{i@}
$. 
 The \emph{consensual language with base} $B$ is defined as \;
$
{\mathcal C}(B) = B^@ \cap \Sigma^*.
$ 
The family of \emph{consensually regular} languages, denoted by  CREG, 
is the collection of
all languages ${\mathcal C}(B)$, such that the base $B$ is regular.    
 \end{definition}
It follows that a CREG  language can be \emph{consensually specified} by a regular expression over $\widetilde{\Sigma}$.
\begin{example}
The  LIP language $L =\{a^n b^n c^n\mid n>0\}$ is consensually specified 
by the base (that we may call a ``consensual regular expression'')
$
\mathring{a}^* a\, \mathring{a}^* \mathring{b}^* b\, \mathring{b}^*\mathring{c}^*c\, \mathring{c}^*
$.
For instance, $aabbcc$ is the (strong) match of 
$\mathring{a}\, a\, \mathring{b}\,   b\, \mathring{c}\  c$ and 
$a\, \mathring{a} \,   b\ \mathring{b}\, c\, \mathring{c}$. 
The commutative closure of $L$ is also in CREG, with base: 
$com\big(a b c\big) \shuffle \mathring{\Sigma}^*$.
\\
 Similarly, the COM-LIP language $L'=com\big((ab)^+\big)=\mathcal{C}(B_1)$, where
$
B_1= com\big(ab\big) $ $ \shuffle \mathring{\Sigma}^* 
$. 
\\
The  COM-LIP language  
$L''=com\big((abb)^+\big)$
is specified  by the base 
$
B_2=com\big( a bb\big) \shuffle \mathring{\Sigma}^* 
$.
\\
The languages $L'\cup L''$  and $L'\cdot  L''$ are  in CREG, but, counter to a naive intuition, they are not specified  by the  bases obtained by composition, respectively, $B_1 \cup B_2$ and $B_1 \, B_2$.
In general $\mathcal{C}(B_1 \cup B_2)\supset \mathcal{C}(B_1)  \cup \mathcal{C}(B_2)$: in the examples,  $\mathcal{C}(B_1 \cup B_2)$ contains also undesirable ``cross-matching'' 
words, such as $ab abb = ab\mathring a\mathring b\mathring b\, @\, \mathring a \mathring b a bb$. 
A systematic compositional technique for obtaining the correct bases for the union and concatenation is the main contribution of this paper. 
\end{example}
\paragraph{Summary of known and relevant  CREG properties.}
{\em Language family comparisons}:  CREG includes the regular
languages, is incomparable with the context-free and deterministic context-free families, is  included within the context-sensitive family,  and it contains
non-SLIP  languages. CREG strictly includes the family of languages accepted by partially-blind multi-counter machines that are deterministic and quasi-real-time, as well as their union~\cite{CrespiSanPietroCIAA13}. 
\\{\em Closure properties:} CREG is is closed under
 marked concatenation, marked iteration, inverse alphabetic homomorphism, reversal,  and intersection and union with regular languages. 
The marked concatenation of two languages $L_1,L_2\subseteq \Sigma^*$ is the language $L_1 \# L_2$, where $\#\not\in\Sigma$, while the marked iteration of 
$L\subseteq \Sigma^*$ is the language $(L \#)^*$. A language family enjoying such properties is known as  a \emph{pre-Abstract Family of Languages} (see, e.g.,~\cite{Salo73}).
A precise characterization of the bases that consensually specify regular languages is in 
\cite{CrespiSanPietroCIAA2012}; an analysis of the reduction in descriptional complexity of the consensual base with respect to the specified regular language is in~\cite{journals/ita/Crespi-ReghizziP11}. 
\\{\em Complexity: }CREG is in NLOGSPACE, i.e., NSPACE$(\log n)$ (often called NL): it can be recognized by a nondeterministic multitape Turing machine working in $\log n$ space. 
The recognizer of CREG languages is a special kind of nondeterministic, real-time multi-counter machine.
\paragraph{Useful notations for consensual languages.}
The following mappings will be used:
\newcommand{\sw}{\textit{switch}}
\begin{center}
$
\begin{array}{ll}
\text{switching }       & \sw: \widetilde\Sigma \to \widetilde \Sigma \text{ where } 
\sw(a) = \mathring a$, $\sw(\mathring a) = a, \text{ for all } a \in \Sigma
\\
\text{marking } &\textit{dot} \,: \widetilde{\Sigma} \to \mathring{\Sigma} \text{ where }
\textit{dot}(x) = x , \text{ if } x \in \mathring\Sigma , 
\text{ and } \textit{dot}(x) =\mathring{a}, \text{ if } x =a\in \Sigma
\\
\text{unmarking } & \textit{undot} \,: \widetilde{\Sigma} \to \Sigma \text{ where } \textit{undot}(a) = \sw(\textit{dot}(a)), \text{ for all } a \in \Sigma.
\end{array}
$
\end{center}
These mappings are naturally  extended to words and languages, e.g., given $x \in\widetilde{\Sigma}^*$, $\sw(x)$ is the word obtained  interchanging 
 $a$ and $\mathring a$ in $x$ (a sort of ``complement''). 
\par
In the remainder of the paper, we assume that each base language is a subset of  $\tilde{\Sigma}^* -\mathring\Sigma^+$, since words in $\mathring\Sigma^+$ are clearly useless in a match. 
\noindent Let $B$, $B'$ be languages included in $\tilde{\Sigma}^+- \mathring\Sigma^+$. We say that $B$ is \emph{unproductive} if ${\mathcal C}(B)= \emptyset$, and that the pair $(B, B')$ is \emph{unmatchable} if 
$B @ B' =\emptyset$.

\section{Consensual specifications composable by union and concatenation}\label{SectDecomposedBases}
Since it is unknown whether the whole CREG family is closed under union and concatenation, we first introduce  a normal form, named decomposed,\footnote{In~\cite{CrespiSanPietroCIAA13},  we introduced the idea of a decomposed form for certain multi-counter machines, but that definition does not work for commutative languages.} of the base languages, which is  convenient to ensure such closure properties.
 Second, we state two further conditions, named joinability and concatenability,  for decomposed forms, and we prove that they, respectively, guarantee closure under union and concatenation. Such results hold for every consensual language, but the difficulty remains to find a systematic method for constructing  base  languages that meets  such conditions. Third, in Sect.~\ref{subSectCongruences} we introduce an implementation of decomposed forms, relying on numerical congruences,  that  will permit us to prove in Sect.~\ref{SectDisciplined} that the ($\cup, \cdot$)-closure of  commutative SLIP languages is in CREG.
%

\begin{definition}[Decomposed form]\label{defBaseInDecompForm}
A base $B\subseteq \tilde{\Sigma}^* -\mathring\Sigma^+$ has the \emph{decomposed form} if there exist a (disjoint) partition of $B$ into two languages, named the \emph{scaffold} $sc$ and the \emph{fill} $fl$ of $B$, 
such that $fl$ is unproductive, and the pair $(sc, sc)$ is unmatchable.
\end{definition}
The names scaffold and fill are meant to convey the idea of an arrangement superposed just once on each word of the base and, respectively, of an optional (but repeatable) component to complete the letters which are dotted in the scaffold. 
Three straightforward  remarks follow.
For every base $B$ there exists a consensually equivalent decomposed base: it suffices to take as scaffold the language  
$\{a \emph{ dot}(y)\mid ay \in B, a\in \Sigma, y\in\widetilde\Sigma^*\}$, and as fill the language  
$\{dot(x) y\mid  x \in\widetilde\Sigma, y\in\widetilde\Sigma^*, xy \in B \}$.
For every  $s \subseteq sc$,  $f \subseteq fl$, the base $s \cup f$ is a decomposed form. 
The scaffold, but not the fill, may include words over $\Sigma$. 
\par
Consider a word  $w \in \mathcal{C}(B)$. Since the fill is unproductive, its match closure  cannot \emph{place} all the letters of $w$ and such  letters  must be placed by the scaffold. Since by definition the match closure of the scaffold alone is the scaffold itself,  the following fundamental lemma immediately holds.

\begin{lemma}\label{lemmaConsLangDecompBase}
If $B= sc \cup fl$ is in decomposed form, as in Def.~\ref{defBaseInDecompForm}, then
$\mathcal{C}(B)=  \left(sc \cup (sc \, @ \, fl^@ \,) \right) \cap \Sigma^*$. 
\end{lemma}

\begin{example}\label{ex-scaffolding}
The table shows  the decomposed bases   of  languages $com\big((ab)^+\big)$  and $com\big((abb)^+\big)$ of Sect.~\ref{sec-consensual}, considering for brevity only the case that the  number of $a$'s is a multiple of 3. Let $L'= com\left(\{a^{3n} b^{3n}\mid n\geq 1\}\right)$, with scaffold $sc'$ and fill $fl'$, and  $L''= 
com\left(\{a^{3n} b^{6n}\mid n\geq 1\}\right)$, with scaffold $sc''$ and fill $fl''$:
\begin{center}
\scalebox{0.9}{%
\small
$
\begin{array}{c|c|c|c}
 & \textit{scaffold} & \textit{fill} & \textit{a strong match}
\\\hline
L' 
&  
(a \mathring a   a  )^+ \, \shuffle \,   ( b \mathring b b )^+
&(\mathring a^3)^* \, \mathring a   a \mathring a \,   (\mathring a^3)^* \,  
\shuffle \, (\mathring b^3)^*\, \mathring b b \mathring b \, (\mathring b^3)^*
&
\begin{array}{lllllll}
& a & b & \mathring a& a&  \mathring b& b \in sc'
\\
@ &\mathring a & \mathring b & a & \mathring a & b &  \mathring b \in fl'
\end{array}
\\
\hline
L'' 
& 
 ( \mathring a  a  a )^+ \, \shuffle \,   (\mathring b b b)^+
& (\mathring a^3)^*\, a \mathring  a  \mathring  a  \, (\mathring a^3)^* \, \shuffle 
 \, (\mathring b^3)^*\, (b \mathring b   \mathring b)^2 \, (\mathring b^3)^*
&
\begin{array}{llllllllll}
& \mathring a& \mathring b &  a & a & b & b & \mathring b & b & b\in sc'
\\
@ & a &  b &  \mathring a & \mathring a & \mathring b & \mathring b & b & \mathring b & \mathring b \in fl'
\end{array}
\end{array}
$
}
\end{center}
Clearly, every word in $sc'$ is unmatchable with every other word in $sc'$, hence $sc'@sc' =\emptyset$. Similarly, every fill is unproductive.
Every word in $L'$ is the match of exactly one word in the scaffold  with  one or more words in the fill. Analogous remarks hold for $L''$. 
\end{example}

Next, imagine to consensually specify two languages by bases in decomposed form $B'=sc' \cup fl'$ and $B''= sc'' \cup fl''$.  
By  imposing additional conditions on the  bases, we obtain two very useful theorems about composition by union and concatenation.


\begin{definition}[Joinability]\label{def-joinable}
Two base languages $B', B''$ in decomposed form
are {\em joinable} if their union  $B' \cup B''$ is decomposed, 
 with scaffold $sc'\cup sc''$ and fill  $fl'\cup fl''$, and the pairs  $(sc', fl'')$ and  $(sc'', fl')$ are unmatchable.
\end{definition}

\begin{theorem}[Union of consensual languages in decomposed form]\label{th-UnionDecomposForms}
Let the base languages $B', B''$ be in decomposed form. If  $B'$ and $B''$ are joinable
then $\mathcal{C}(B') \cup \mathcal{C}(B'') =  \mathcal{C}\left(B' \cup B''\right)$.
\end{theorem} 
\begin{proof}
It suffices to prove the inclusion   $\mathcal{C}(B' \cup B'') \subseteq  \mathcal{C}(B') \cup \mathcal{C}(B'') $,  since the opposite inclusion 
is obvious by Def.~\ref{defConsensualLang}. 
Let $x \in \mathcal{C}(B) $. Since $B$ is decomposed, by Lemma~\ref{lemmaConsLangDecompBase}
it must be either $x \in sc @
fl^@$ or $x \in sc$. In the latter case,  $x$ is in $B'$ or in $B''$, and the inclusion follows.  
In the former case, there exist $n\ge 2$ words  $w_1, w_2 \dots, w_n$, with $n\le|x|$, 
$w_1 \in sc$, $w_2, \dots, w_n\in fl$ and $w_1 @ w_2 @ \dots  @ w_n = x$.
We claim that 
either $w_1 \in sc'$ and every other $w_i\in B'$,  
or $w_1 \in sc''$ and every other $w_i \in B''$, from which the thesis follows. 
Assume $w_1\in sc'$ (the case $w_1\in sc''$
is symmetrical).
If there exists $j$, $2\le j\le n$, such that $w_j\in
fl''$ (with $w_j\not\in\mathring\Sigma^+)$, then $sc'@fl''$ is not empty (it includes at least $w_1 @ w_j$), a contradiction with the hypothesis that 
$B'$ and $B''$ are joinable.
\end{proof}
\par
\begin{example}\label{ex-union}
Returning to Ex.~\ref{ex-scaffolding}, we check that the two bases are joinable. The union of the bases is  in decomposed form:
$fl'\cup fl''$ is unproductive (because letters at positions 3, 6, \ldots cannot be placed); the  pair $(sc',sc'')$ is unmatchable, hence also $(sc'\cup sc'',sc'\cup sc'')$ is unmatchable. 
Moreover, $(sc', fl'')$, and  $(sc'', fl')$ are unmatchable.
Therefore   $L' \cup L'' = \mathcal{C}(sc' \cup sc'' \cup fl' \cup fl'')$.
\end{example}
\par
For concatenation, a similar, though more involved, reasoning  requires a new technical definition.
\begin{definition}[Dot-product $\odot$ and concatenability]\label{defConcatenability}
Let $B', B''$ be in decomposed  form, and  define their \emph{dot-product}   as
$B' \odot B''= (sc'\cdot sc'') \cup fl' \cup fl''$. 
$B'$ and $B''$ are \emph{concatenable} if $B' \odot B''$ is in decomposed form, with scaffold $sc'\cdot sc''$ and fill $fl' \cup fl''$, 
and
the next two clauses hold for all words $w',w''\in\widetilde\Sigma^+,\, y' \in sc',\, y''\in sc''$:
\begin{eqnarray}
\label{equ1Concatenable}
\exists x' \in fl' : w'=x'\cdot dot(y'') \; \land \;\; x'@y' \text{ is defined } 
\text{ if, and only if, }
 w'\in fl' \wedge \; w' @ y'\cdot y'' \text{ is defined } 
\\
\label{equ2Concatenable}
\exists x'' \in fl'' : w''=dot(y') \cdot x'' \; \land \;\; x''@y'' \text{ is defined }  
\text{ if, and only if, } \;  w''\in fl'' \wedge  \; w'' @ y'\cdot y''  \text{ is defined}
\end{eqnarray}
\end{definition}
The two clauses are symmetrical. In loose terms, Clause~\eqref{equ1Concatenable} says that the fill $fl'$ contains a word $w'$ that matches $y'y''$, if, and only if, the word has a prefix $x'$ , also in $fl'$, which matches  $y'$, hence it is aligned with the point of concatenation. Therefore,  the match  $w' @ y'\cdot y''$ does not produce a word that is illegal for $\mathcal{C}(B') \cdot \mathcal{C}(B'')$. This reasoning is formalized and proved next.

\begin{theorem}[Concatenation of consensual languages in decomposed form]\label{th-ConcatDecomposForms}
Let the bases  $B', B''$ be in decomposed  form. 
 If  $B', B''$ are concatenable, then $\mathcal{C}(B') \cdot \mathcal{C}(B'') =  \mathcal{C}\left(B'\odot B''\right)$.

\end{theorem}
\begin{proof}
Let $B=B'\odot B''$.\\ 
$\textit{Case } \mathcal{C}(B')  \cdot\, \mathcal{C}(B'') \subseteq \mathcal{C}(B).
$ 
If $x \in \mathcal{C}(B')  \cdot
\mathcal{C}(B'')$, then $x=x'x''$ with $x' \in \mathcal{C}(B')$, $x'' \in \mathcal{C}(B'')$.
Hence, $x'$ is the strong match of one $w'\in sc'$ (resp.  $w''\in sc'' $) with
$n\ge 0$ words $w'_1, \dots, w'_n\in fl'  \subseteq  fl $; analogously, $x''$
is the strong match of one $w''\in sc'' $  with $m\ge 0$ words $w''_1, \dots w''_m\in fl'' $. 
By definition of concatenability, since for $1\le i\le n$, 
every word $w'_i$ is in $fl'$, then also all words $w'_1\cdot dot(w''), w'_2\cdot dot(w''), \dots $ are in  $fl'$, hence also in $fl$. 
Similarly, also  $dot(w'') \cdot  w''_1, \dots dot(w'') \cdot w'_n$ are in $fl''$.
Since $w'\cdot w''$ is in $sc' sc''$,  it is possible to define a strong match
yielding $x' x'' = x$, namely,
\[x=w'w'' @ \left(w'_1\cdot dot(w'')\right) @ \left(w'_2\cdot dot(w'')\right) @
\dots
\left(dot(w') \cdot w''_1) @ (dot(w') \cdot w''_2 \right) @ \dots\]
that  is the concatenation of 
$w'@ w'_1@ \dots @w'_n  = x'$ with 
$w'' @ w''_1 @\dots @w''_m =x''$.
\par
$\textit{Case }  \mathcal{C}(B) \subseteq \mathcal{C}(B')  \cdot\, \mathcal{C}(B'').
$
Let $x \in \mathcal{C}(B)$. Then 
there exist $n\ge 1$ words  $w_1, w_2, \dots, w_n$, with $n\le|x|$, such that 
$w_1 @ w_2 @ \dots @ w_n = x$, $w_1 \in sc' \cdot sc''$ and   
$w_2, \dots, w_n\in  fl' \cup fl'' $.
By definition,  $w_1$ can be decomposed into $w_1 = w'_1 w'_2$ for some 
$w'_1\in sc' , w''_2\in sc'' $.  Let  $q= |w'_1|$.
Assume, by contradiction, that $x \not\in \mathcal{C}(B') \cdot\,
\mathcal{C}(B'')$. 
Since $x$ is the match of  word $w_1= w'_1 w'_2$ and words in $fl' \cup fl''$, 
the only possibility for $w$ not being in 
$\mathcal{C}(B')  \cdot\, \mathcal{C}(B'')$ is that 
there exists $j, 2\le j \le n$, such that:
\begin{enumerate}
 \item 
$w_j \in fl'$, and  the substring $w_j(1,q) \not\in fl'$, or 
\item
$w_j \in  fl'' $, and the substring
$w_j\left(q+1,|x|\right) \not\in  fl''
$.
\end{enumerate}
We consider only Case (1) since the other is symmetrical. 
Since $w_j \in fl'$ and $w_j @ w'_1 w''_1$ is defined, then,  by definition of concatenability, there exists $x'\in fl'$ such that
 $w_j= x'\cdot dot(w''_1)$, i.e., $w_j(1,q) = x'$, a contradiction with the assumption of Case (1). 
\end{proof}

\begin{example}\label{ex-catenation}
Consider again Ex.~\ref{ex-scaffolding}.  
It is easy to  check that the pair  $(sc'\cdot sc'',\,sc'\cdot sc'')$  is unmatchable, for the same reason that $(sc',sc'')$ is unmatchable. Then, we check that the  bases  $sc' \cup fl'$ and  $sc'' \cup fl''$ are concatenable. We only discuss the case of Clause \eqref{equ1Concatenable} since  Clause \eqref{equ2Concatenable} is symmetrical. 
Let $w' \in \widetilde\Sigma^+$, $y'\in sc'$, $fl''\in sc''$. If there exists $x'\in fl'$ such that $w'=x'dot(y'')$, then obviously both $w'\in fl'$ and  $w'@ y'\cdot y''$ are defined.
\\
For the converse case, assume that $w' \in fl'$ and $w'@ y'\cdot y''$ is defined. 
Consider the projections $\alpha=\pi_{\widetilde{a}}(w')$,  $\alpha'=\pi_{\widetilde{a}}(y')\in (a \mathring aa)^+$ and $\alpha''=\pi_{\widetilde{a}}(y'')\in (\mathring a aa)^+$.
Then $\alpha \in (\mathring a\mathring a \mathring a)^* \mathring aa\mathring a(\mathring a\mathring a \mathring a)^*$. Since $w'@ y'\cdot y''$ is defined, the factor $\mathring aa\mathring a$ of $\alpha$ must be 
matched with a factor of $\alpha'\alpha''$: by its form and alignment, the only possibility is that it is matched with a factor of $\alpha'$. Hence, $\alpha$ has the form  
 $(\mathring a\mathring a \mathring a)^*  \mathring aa\mathring a (\mathring a\mathring a \mathring a)^*  dot(\alpha'')$. We omit the  analogous reasoning for the projections on $b$. Since $w'@ y'\cdot y''$ is defined, then $w'$ must have the form
$x' \cdot dot(y'')$ for some $x'\in fl'$. 
Therefore  $L' \cdot L'' = \mathcal{C}(sc' \cdot sc'' \cup fl' \cup fl'')$. For instance 
\[
a^3b^3 a^3 b^6 = 
\begin{array}{ll}
a \mathring a   a b \mathring b   b \cdot \mathring a  a  a \mathring b  b  b \mathring b  b  b \;@
\\
\mathring a  a  \mathring a  \mathring b  b \mathring b \cdot \mathring a \mathring a \mathring a \mathring b  \mathring b \mathring b \mathring b \mathring b \mathring b  \;@
\\
\mathring a \mathring a  \mathring a  \mathring b \mathring b \mathring b \cdot  a \mathring a \mathring a  b  \mathring b \mathring b  b \mathring b \mathring b  
\end{array}
\]
\end{example}
\noindent This example relies on  a numerical congruence with module 3 for positioning the dotted and undotted letters. We shall see how to generalize this approach to handle words of any congruence class (with respect to the length of the projections on each letter). The generalization  will carry the cost of taking larger values for the congruence module.
\par 
Incidentally,  we observe that the theorems of this section may have a more general use than for commutative languages. Moreover, the theorems do not require the base languages to be regular; in fact, Def.~\ref{defConsensualLang} applies as well to non-regular bases (as  a matter of fact~\cite{CrespiSanPietroCIAA2012} studies  context-free/sensitive bases). 

\subsection{A Decomposed Form Relying on Congruences}\label{subSectCongruences}
Having stated some  sufficient  conditions for ensuring that the union/concatenation of two consensual languages can be obtained by composing (as described by Th. 
\ref{th-UnionDecomposForms} and Th.~\ref{th-ConcatDecomposForms}) the corresponding base languages,
 we design a decomposed form, suitable for supporting joinability and concatenability,  that uses module arithmetic for  assigning 
the positions to the dotted and  undotted letters  within a word $w$ 
over $\widetilde\Sigma$; the preceding examples offered some intuition for the next formal developments.\footnote{As said, similar ideas have been used for a different language family in~\cite{CrespiSanPietroCIAA13} and have been sketched for COM-SLIP languages in our communication~\cite{conf/ictcs/Crespi-Reghizzi13}.} 
Loosely speaking, each decomposed base language is ``personalized'' by a  sort of unique pattern of dotted/undotted letters, such that, when we want to   unite or concatenate two languages,  the match of two words with different patterns  is undefined,  thus ensuring that the union or catenation of the two  decomposed bases specifies the intended language composition. 
\par
For every $a\in\Sigma$, consider the 
projection of $w$ on $\widetilde{a}= \{a,\mathring a\}$ and, in there, the numbered positions of each $a$ and $\mathring a$. 
Let $m$ be an integer.
By prescribing that for each  base language, each undotted letter $a$ may 
only occur in  positions $j$ characterized by a specified value of the congruence   
$j \mod m$, we  make the bases decomposed. 
We need a new definition.

\begin{definition}[Slots and modules]\label{defR_m(a)R_m}\label{def-scaffolding}
Let $m>3$, called \emph{module}, be an even number. 
Let $R\subseteq \{1, \dots, (m/2 - 1)\}$ be a nonempty set, called a \emph{set of slots of module $m$}.
For every $a\in \Sigma$, define a  finite language  $R_m(a)\subset 
\tilde a^m$, where only  positions 1 and $r+1$ are dotted: 
\begin{equation}	R_m(a)  = \{\mathring a\, a^{r-1} \mathring a\,  a^{m-r-1} \mid r \in R\}
\end{equation}
The disjoint  regular  languages $\scaffolding{R}_{m}, \filling{R}_{m}\widetilde\Sigma^*$ are defined as:
\begin{eqnarray}
 \label{eq.sc-Rm} 
\scaffolding{R}_{m} &=& \left\{x  \mid  \forall a \in \Sigma, \pi_{\widetilde a}(x) \in 
\left(R_m(a) \cup a\right)^* \right\}
\\ \label{eq.fl-Rm}  
\filling{R}_{m} &=&  \sw(\scaffolding{R}_{m}) -\mathring\Sigma^*.
\end{eqnarray}
\end{definition}

The definition of $\filling{R}_m$ is clearly equivalent to $\left\{x  \mid  \forall a \in \Sigma, \pi_{\widetilde a}(x) \in 
\left(\sw(R_m(a)) \cup \mathring a\right)^* \right\} -\mathring\Sigma^*$. It is fairly obvious that $\mathcal{C}(B) = \Sigma^+$, since $\Sigma^+\subseteq \scaffolding{R}_m$.
Also, $\scaffolding{R}_{m} @ \scaffolding{R}_{m} = \emptyset$ and $\filling{R}_m$ is unproductive. The following lemma is also obvious.
\begin{lemma}\label{lm-scaffolding}
For all even numbers $m>3$ and non-empty sets $R$ of slots of  module $m$, 
every base $E\subseteq \scaffolding{R}_{m} \cup \filling{R}_{m}$ is in decomposed form, with  scaffold:  $E \cap \scaffolding{R}_{m}$  and  fill:  
$E\cap \filling{R}_{m}$.  
\end{lemma}

\begin{example}
Let $m = 6, R= \{1, 2\}$ and $\Sigma=\{a,b\}$. Then
\[
\begin{array}{lll}
R_6(a) &=& \{\mathring a \mathring a  a  a a a , \,  \mathring a a \mathring a    a  a a\}
\\
\scaffolding{R}_6 & = & ({\bf\mathring a \mathring a  a  a a a} \cup {\bf\mathring a a \mathring a    a  a a} \cup a)^* \shuffle 
({\bf\mathring b \mathring b  b  b   b b} \cup  {\bf\mathring b b \mathring b    b  b b} \cup b)^*
 \\
\filling{R}_6 & = &  \left( ({\bf  a  a \mathring a \mathring a \mathring a \mathring a} \cup  {\bf a \mathring a  a \mathring   a \mathring a \mathring a }\cup \mathring a)^* \shuffle \; 
({\bf  b  b\mathring  b\mathring   b\mathring   b\mathring b } \cup   {\bf b\mathring b  b  \mathring  b \mathring b\mathring b} \cup \mathring b)^*\right) - \{\mathring a, \mathring b\}^*
\end{array}
\]
For clarity, in this example the characters in $\scaffolding{R}_6$ and in $\filling{R}_6$,  belonging to factors in $R_6(a), R_6(b)$, or $\sw(R_6(a)), \sw(R_6(b))$ respectively, 
are in bold. 
Examples of words in $\mathcal{C}(B)$ are:
\[
a^6b^6 \in \scaffolding{R}_6,\; \text{ also }  a^6b^6 = 
\begin{array}{l|l}
\bf{\mathring a a \mathring a    a  a a \mathring b b \mathring b    b  b b}\;  @ & \text{in }\scaffolding{R}_6
\\
 {\bf 
a \mathring a a    \mathring a  \mathring a\mathring a b \mathring b b    \mathring b  \mathring b\mathring b 
} 
& \text{in }\filling {R}_6
\end{array}
\]
\[
a^9b^8 \in \scaffolding{R}_6,\; \text{ also }  a^9b^8 = 
\begin{array}{l|l}
{\bf\mathring a a \mathring a    a  a a} aaa {\bf \mathring b b \mathring b    b  b b}bb \;  @ & \text{in }\scaffolding{R}_6
\\
{\bf a \mathring a a    \mathring a  \mathring a\mathring a} \mathring a\mathring a\mathring a {\bf b \mathring b b    \mathring b  \mathring b\mathring b } \mathring b\mathring b & \text{in }\filling {R}_6
\end{array}
\]

\[
(ab)^4aaabb \in \scaffolding{R}_6,\; \text{ also } (ab)^4aaabb = 
\begin{array}{l|l}
 {\bf\mathring a \mathring b a b \mathring a \mathring b   a  b a a} a {\bf      b b} @ & \text{in }\scaffolding{R}_6
\\
{\bf a b \mathring a \mathring b a b   \mathring a  \mathring b \mathring a \mathring a} \mathring a   {\bf \mathring b \mathring b} & \text{in }\filling {R}_6
\end{array}
\]
\end{example}
\par
To ensure that a base, included in $\scaffolding{R}_{m} \cup \filling{R}_{m}$, can be used when two such languages are concatenated,  we need the next simple concept.
\begin{definition}[Shiftability]\label{def-shift} 
A  language $R\subseteq \tilde{\Sigma}^*$ is {\em shiftable} if  $R= \mathring\Sigma^*\, R\, \mathring\Sigma^* $.
\end{definition}
\noindent This means that any word in $R$ remains legal, when  it is padded to the left/right  with any dotted words.
\par
Next we show that by taking disjoint sets of slots over the same module, we obtain two bases that are joinable; if, in addition, the fills are shiftable, the condition for concatenability is satisfied.

\begin{theorem}\label{th-join-concat}
Let $m>3$ and let $R',R''$ be two disjoint sets of slots of module $m$, and let 
$E'\subseteq \scaffolding{R'}_{m}\cup \filling{R'}_m$ and $E''\subseteq\scaffolding{R''}_{m}\cup \filling{R''}_m$ be two bases. Then:
\begin{itemize}
 \item $E'$ and $E''$ are joinable;
\item if the fills of $E'$ and $E''$ are shiftable, then the fills of $E'\cup E''$ and $E'\odot E''$ are also shiftable, and $E'$ and $E''$ are concatenable.
\end{itemize}
\end{theorem}
\begin{proof}
Let $R= R'\cup R''$. Bases $E'$ and $E''$ are in decomposed form by Lm.~\ref{lm-scaffolding}. 
Also $E'\cup E''$ and $E'\odot E''$ are in decomposed form, since they are both subsets of $\scaffolding{R}_{m} \cup \filling{R}$.
\\
\noindent {\bf Part (1): } 
To show that $E'$ and $E''$ are joinable, we only need to prove that $(\filling{R''}_m, \scaffolding{R'}_{m})$ is unmatchable
(the case $(\filling{R'}_m, \scaffolding{R''}_{m})$ being unmatchable is symmetrical). 
By contradiction, assume that there exist $x \in \filling{R''}_m$ and $y \in\scaffolding{R'}_{m}$ such that $x@y$ is defined. 
Let $a\in\Sigma$ be a letter occurring in $x\not\in\mathring\Sigma^+$ and consider the projection $\alpha= \pi_{\widetilde{a}}(x)$. 
By definition of $\filling{R''}_m$, there exist a position $q$ of $\alpha$ and a value $r\in R''$
such that $\alpha(q)=\alpha(q+r'')=a$.
Then, there exists $\alpha'\in \pi_{\widetilde{a}}(y)$ such that
$\alpha@\alpha'$ is defined. But in $\alpha'$ for all positions $p$, $1\le p \le |\alpha'|$, if $\alpha'(p) =\mathring a$ then $\alpha'(p+r')= a$ for all $r'\not\in R'$. 
Therefore, if $p=q$ then $\alpha(p+r) = \alpha'(p+r) = a$, which is impossible by definition of matching. 
The same argument could be applied to show that also the other two pairs are unmatchable.
\\ \noindent {\bf Part (2): } Define as $\filling{E'},\scaffolding{E'}$ and as $\filling{E''},\scaffolding{E''}$ the fills and the scaffolds of $E'$ and $E''$, respectively.
 If $\filling{E'}$ and $\filling{E''}$ are shiftable, then also the fill $\filling{E'}\cup\filling{E''}$ of both $E'\cup E''$ and $E'\odot E''$ is shiftable, since the union of two shiftable languages is shiftable. 
We now prove that in this case $E',E''$ are also concatenable.
Let $w'\in\filling{E'},y'\in\scaffolding{E'},y''\in\scaffolding{E''}$.
If there exists $x'\in\filling{E'}$ such that $x'@y'$ is defined and $w' = x'dot(y'')$, then it is obvious that $w'\in\filling{E'}=\mathring\Sigma^*\filling{E'}\mathring\Sigma^*$ and 
that $w'@(y'\cdot y'')$ is defined. 
We are left to show that: 
\begin{equation} 
\label{(*)}
 \text{if } w'@(y'\cdot y'') \text{ is defined then }  \exists x'\in\filling{E'} 
\text{such that } w' = x' dot(y'') \text{ and } x'@y' \text{ is defined.}
\end{equation}
The proof of Claim \eqref{(*)} requires another technical definition.
Given a set $R$ of slots with module $m$, for $a \in \Sigma$, for every  $\alpha \in \pi_{\widetilde a}(\scaffolding{R}_{m})$ 
a {\em restarting point} for projection $\alpha$ is 
a position $i$, $1\le i \le |\alpha|-m$, 
such that $\alpha(i,i+m-1) \in R_m(a)$. Hence, at $i$ there is a factor in $R_m(a)$. 
A symmetrical definition holds if  $\alpha \in \pi_{\widetilde a}(\filling{R}_{m})$:  
factor  
$\alpha(i,i+m-1) \in switch(R_m(a))$.
A restarting point always exists for all $\alpha \in \pi_{\widetilde a}(\scaffolding{R}_{m})$ 
or $\alpha\in\pi_{\widetilde a}(\filling{R}_{m})$, 
provided that $\alpha\not\in\Sigma^+$. 
We claim that 
if $s\in \scaffolding{R}_{m}, f\in \filling{\hat{R}}_m$
for some (possibly equal) sets of slots $R, \hat{R}$ with module $m$,  and the match $s @ f$ is defined,  
 then both the following conditions hold:
\begin{align}
\label{(I)}
 &R\cap  \hat{R} \neq \emptyset, 
\\
 \label{(II)}
&\forall  a\in \Sigma, \text{ the set of restarting points for } \pi_{\widetilde a}(f) \text{ is included in the set of restarting points 
for } \pi_{\widetilde a}(s).
\end{align}
Since $f\not\in \mathring\Sigma^*$, there exists at least one $a \in \Sigma$ such that 
$\pi_{\widetilde a}(f)$  has a factor in $switch( \hat{R}_m(a))$
i.e., 
 there exists a restarting point $p$ for $\pi_{\widetilde a}(f)$.
For brevity, let $\alpha = \pi_{\widetilde a}(f)$. 
Hence, 
 $1\le p \le |\alpha|-m$.
Therefore, there exists $r\in \hat{R}$ such that $\alpha(p) = \alpha(p+r) =
a$.
Consider now $\beta =  \pi_{\widetilde a}(s)$. Since $s @ f$ was assumed to be defined, $\beta(p) = \beta(p+r) = \mathring a$.
By definition of $\scaffolding{R}_m$, 
$\beta  \in (R_m(a) \cup a )^*$. 
\\
There are two possibilities: either $p$ is a restarting point also for $\beta$, hence $r\in R$ and the above claims follow, 
or $p$ is not a restarting point for $\beta$. The latter case is however impossible. In fact, in this case $p+r$ would be a restarting point for $\beta$, 
because of the form of $R_m(a)$.  Therefore, since $\beta(p) = \mathring a$, there would be a restarting point also at position $p-r'$, for some $r'\in R$. 
However, both $r$, $r'$, by definition, are smaller than $m/2$, therefore $2 \le r+r' \le m -2$. Hence, the restarting point at $p-r'$ would be at a distance
less than $m$ from the restarting point at $p+r$, which is impossible by definition of $R_m(a)$.
\par
We  prove Claim \eqref{(*)} to finish. For every $a \in \Sigma$, let $q'_a = |\pi_{\widetilde{a}}(y')|$, and let $q''_a =|\pi_{\widetilde{a}}(y')|$. 
Consider the rightmost restarting point $p_a$ for $\pi_{\widetilde{a}}(w')$. By definition of $\filling{E'}$, 
there exists $r'\in R'$ such that $\pi_{\widetilde{a}}(w')(p_a,p_a+m)= a \mathring a^{r'-1} a \mathring a^{m-r'-1}$. 
By Claim \eqref{(II)},  
$p_a$ is also a restarting point for $\pi_{\widetilde{a}}(y'\cdot y'')$: 
there exists $r\in R'\cup R''$ such that $\pi_{\widetilde{a}}(y'y'')(p_a,p_a+m)= \mathring a a^{r-1}\mathring a a^{m-r-1}$.
We claim that $p_a\le q'_a$. In fact, if $p_a>q_a$, then $p_a$ must be a restarting point for $y''$, hence $r\in R''$: but $r=r'$, a contradiction
with the hypothesis that $R'\cap R''=\emptyset$.  
If $p_a\le q'_a$ then $p_a$ must be a restarting point for $\pi_{\widetilde{a}}(y')$, hence $r=r'$ and actually $p_a \le q_a-m$. 
Since $p_a$ is the rightmost restarting point, $\pi_{\widetilde{a}}(w')(p_a-m+1,q'_a+q''_a) \in \mathring\Sigma^+$. Choose $x'$ to be the prefix of $w'$ such that such that $w'=x'dot(y'')$.
\end{proof}

\section{Commutative SLIP languages and their \texorpdfstring{$(\cup, \cdot)$}{union-concatenation}-closure}\label{SectDisciplined}
This section proves the main result: 
\begin{theorem}[Closure under union and concatenation]\label{th-main}
The family \C is strictly included in the family of consensually regular languages: $\C \subset \text{CREG}$. 
\end{theorem}
Every language in \C can be defined by an expression that combines finitely many COM-SLIP languages,  using union and concatenation; since COM-SLIP is the finite union of COM-LIP languages, we may assume that 
the expression includes only COM-LIP, rather than COM-SLIP, languages.
\par
In the sequel, we prove that  every COM-LIP language can be consensually defined in a decomposed form such  that it permits to satisfy the additional assumptions needed for union and concatenation,
hence all \C languages are in CREG. 
\paragraph{Decomposed form for COM-LIP languages}\label{SectDisciplinedFormLIP}
To expedite handling the constant terms of LIP systems, we introduce a new operation \emph{append} that combines a language and a commutative  language, the latter penetrating into the former. 
\begin{definition}[Appending]\label{def-append}
Let  $B$ be a  language over the double alphabet  $\widetilde\Sigma$.
For $a\in \Sigma$, define the (unique)  factorization 
\begin{center}
$B = B_{\widetilde{a}}\cdot B_{\widetilde\Sigma - \widetilde{a}}$
\end{center}
 where 
$B_{\widetilde{a}}\subseteq\widetilde\Sigma^*\cdot {\widetilde{a}}$ and $B_{\widetilde\Sigma - \widetilde{a}}\subseteq\left(\widetilde\Sigma-\widetilde{a}\right)^*$ are  languages, resp. ending by $\widetilde{a}$, and not using the letters  $a, \mathring a$. If neither $a$ nor $\mathring a$ occurs in $B$, let $B_{\widetilde{a}}=\varepsilon$.
Let $A\subseteq a^+$; we define the  operation, named \emph{appending  $A$ to $B$},   as follows:
\begin{center}
$B \lhd A = B_{\widetilde{a}} \cdot (B_{\widetilde\Sigma - \widetilde{a}} \shuffle A)$.
\end{center} 
Given a  commutative language $F \subseteq \Sigma^*$,  $\Sigma= \{a_1, \dots, a_k\}$,
the iterative application of  the previous operation to every letter of the alphabet (in any order) defines the operation, named \emph{letter-by-letter appending} $F$ to $B$, as:
\begin{center}
$B\lhd F = \left(\dots (B \lhd {\pi_{a_1}(F)}) \lhd {\pi_{a_2}(F)} ) \dots \right) \lhd {\pi_{a_k}(F)}$.
\end{center}
\end{definition}
To illustrate, we  compute:
\begin{gather*}
\{\mathring a b \mathring a \mathring b \} \lhd \{ac, ca \} 
= \left( \{\mathring a b \mathring a \mathring b \}\lhd \pi_a\{ac, ca \}\right) \lhd \pi_c\{ac, ca \}=\\
= \left( \{\mathring a b \mathring a \mathring b \}\lhd \{a \}\right) \lhd \{c \} =
 \left( \{\mathring a b \mathring a  \} (\mathring b  \shuffle a) \right) \lhd \{c \}=\\
=\{\mathring a b \mathring a \mathring b a, \mathring a b \mathring a a \mathring b  \}\lhd \{c \}=
\{\mathring a b \mathring a \mathring b a, \mathring a b \mathring a a \mathring b  \}\shuffle  \{c \}
\end{gather*} 

In the remainder of the Section, let $L$ be a COM-LIP language over $\Sigma=\{a_1, \ldots, a_k\}$, $k>0$, defined by  constant 
$\vec{c}$   and  periods 
$\mathcal{P} = 
\left \{\vec{p}^{(1)}, \ldots,
\vec{p}^{(q)}\right\}$, for some $q>0$, with the condition that 
for every $\vec{p}\in \mathcal P$, every component $p_i$ is even.

The next definition introduces some sets, called $X, Y,W$,  to define the COM-LIP language $L$ with a base $D$ in decomposed form.
The assumption on each $p_i$ being even will be lifted when defining COM-SLIP languages.

\begin{definition}\label{def-decomposedCOMLIP}
For all even integers $m\ge 4$, and for all sets of slots $R$ of the form $\{r\}$ with $0<r<m/2$, 
define the regular languages $X, Y, D\subseteq \widetilde\Sigma^*$ and  the finite commutative language $W\subseteq \Sigma^*$,
as follows:
\begin{equation}
X= \bigcup_{\vec{p}\in \mathcal{P}}\{x\in \filling{R}_m 
\mid 
\Psi(\pi_\Sigma(x)) = \vec{p}
\}
\label{eqX}
\end{equation}
\begin{equation}
Y = (R_m(a_1))^* \shuffle \dots \shuffle (R_m(a_k))^* 
\label{eqY} 
\end{equation}
\begin{equation}
\Psi(W) = 
\left\{ \vec{c} +h_1 \cdot\vec{p}^{(1)}+\ldots +  h_q \cdot
\vec{p}^{(q)} \mid 0 \leq h_1, \ldots, h_q <m/2 \right\}.
\label{eqW}
\end{equation}
\begin{equation}
\label{eqDisciplBaseLang}
D = 
 X 
\;\cup\,\left( Y \lhd W \right)
\end{equation}
\end{definition}

It is obvious that $X\subseteq \filling{R_m}$. To see that $Y \lhd W \subseteq \scaffolding{R_m}$, 
we first describe  relevant features of the formulae.
By Eq.~\eqref{eqW},  $W$ is the finite commutative language having as Parikh image 
the linear subspace included between $\vec{c}$ and $ \vec{c} + (m/2-1) \vec{p}^{(1)}+\ldots + (m/2-1) \vec{p}^{(q)}$.
For each  $a_i$, the projection on $a_i$ of a word in $Y \lhd W $ ends 
with a tail of  undotted $a_i$'s defined by  Eq.~\eqref{eqW}. While the projection on $a_i$  of $\scaffolding{R}_m$ has necessarily length multiple of  $m$, 
the tail does not need to comply with such constraint, thus allowing, in principle, the language $Y \lhd W$ to contain words whose projections on $a_i$ 
has any length greater or equal to $c_i$ (within the specified subspace).
The following lemma is immediate:
\begin{lemma} 
Let $X,Y,W,D$ as in Def.~\ref{def-decomposedCOMLIP}. 
Then,  $D$ is a decomposed base included in $\scaffolding{R_m} \cup \filling{R_m}$, with $Y \lhd W\subseteq \scaffolding{R_m}$ being the scaffold and $X \subseteq \filling{R_m}$ being the fill; 
moreover, the fill of $D$ is shiftable, i.e., $X = \mathring{\Sigma}^* X \mathring{\Sigma}^*$. 
\end{lemma}

\begin{example}
Consider the language $L''_{even} = com\big((a^{2} b^{4})^*\big)$ having  the period $p_{a} = 2, p_{b} = 4$ and null constant. 
Notice that to obtain language $com\big((ab^{2})^*\big)$, it is enough to apply union to $L''_{even}$ and to the language 
 $L''_{odd} = com\left(abb(a^{2} b^{4})^*\right)$, which can be defined with 
the same period $p_{a} = 2, p_{b} = 4$, and with constant $c_a = 1, c_b =2$.
If module $m=6$ and  set of slots  $R=\{ 2\}$ then $R_6(a) = \mathring a a\mathring  a { a}^3,\, R_6(b) = \mathring b b\mathring  b {b}^3$. Also, 
$\filling{R}_6= \left( \left(a \mathring a a {\mathring a}^3  \cup \mathring a \right)^* \;\shuffle \;\left(b \mathring b b {\mathring b}^3  \cup \mathring b \right)^* \right) - \{ \mathring a, \mathring b \}^*
$. Let
\begin{eqnarray*}
X &=& \{x \in \filling{R}_6  \mid \Psi\left(\pi_{\{a,b\}}(x) \right) = (2,4) \} 
\\
&=& \left( {\mathring  a}^* \cdot a \mathring a a {\mathring a}^3\cdot {\mathring  a}^*\right) \;\shuffle \; 
\left({\mathring  b}^* \cdot b \mathring b b {\mathring b}^3\cdot {\mathring  b}^*\cdot b \mathring b b {\mathring b}^3\cdot {\mathring  b}^*\right)
\\
Y &=& \left(R_6(a)\right)^* \shuffle  \left(R_6(b)\right)^* = 
\left(\mathring a a\mathring  a { a}^3 \right)^* \shuffle \left(\mathring b b\mathring  b {b}^3 \right)^*
\end{eqnarray*}
 Both $X$ and $Y$ satisfy Def.~\ref{def-decomposedCOMLIP}.
To complete the base of  language $L''_{even}$, we define 
\[
 W = \bigcup_{0\le i \le 2} com\left(a^{2i} b^{4i}\right)
\] 
The fill  $\{ \mathring a, \mathring b \}^* X \{ \mathring a, \mathring b \}^*$ and the  scaffold $Y \lhd W$ are a decomposed form for  $L''_{even}$.
Similarly, to define  $L''_{odd}$, we have to define  the sets $X', Y', W'$;  for  $X', Y'$ we select as set of slots $R'=\{1\}$, which satisfies
 $R'\cap R = \emptyset$. At last,  $W' = \bigcup_{0\le i \le 2} com\left(abb a^{2i} b^{4i}\right)$.
\end{example}
The important property of the language in Eq.~\eqref{eqX} is stated  next.
\begin{lemma}\label{lm-LinearCombinOfPeriods} 
\begin{enumerate}
\item 
For all $n>0$, for every $
 u \in  
X^{n@}
$
there exist $q\ge 1$ integers $n_1, \ldots, n_q \geq 0$ with $n=n_1 + \ldots + n_q$ such that 
\[\Psi\left( \pi_\Sigma\, (u) \right) = n_1 \cdot \vec{p}^{(1)} + \ldots + n_q \cdot \vec{p}^{(q)}.\]
\item For all $n, n_1, \ldots, n_q \geq 0$, with  $n_1 + \ldots + n_q = n$ , 
if 
\[u \in \filling{R}_m \;\text{ and }\; \Psi\left( \pi_\Sigma\, (u) \right) = n_1 \cdot \vec{p}^{(1)} + \ldots + n_q \cdot \vec{p}^{(q)}\]
then $ u \in   X ^{n@}$.
\end{enumerate}
\end{lemma}
\begin{proof}
Part (1). By definition of $X$, if $x \in X $, then 
there exists $\vec{p}^j\in \mathcal{P}$, $1\le j \le q$, such that $\Psi\left(\pi_{\Sigma(x)}\right) = \vec{p}^j$.
By definition of  match closure, there exists $n>0$ words $x_1, \dots x_n \in 
X$
such that $u = x_1 @ x_2 @ \dots @ x_n$. 
Then, for all $1\le i \le n$, $\Psi(\pi_{\Sigma(x_i)} = \vec{p}^{j_i}$ for some $j_i$, with $1\le j_i \le q$. Hence, 
$\Psi\left( \pi_\Sigma\, (u) \right) = \sum_{1\le i \le n} \Psi(\pi_{\Sigma(x_i)})$, from which the thesis follows immediately.   
Part (2).
By definition of $X$, for every vector $\vec{p}^j$, $1\le j \le q$,
language 
$X$
includes all words $x$ of $\filling{R}_m$ such that $\Psi(\pi_{\Sigma(x)}) = \vec{p^j}$.
Hence, one can always select  $n_1$ words $x^{[1]}_1, \dots, x^{[1]}_{n_1}\in X$,
$n_2 $  words $x^{[2]}_1, \dots , x^{[2]}_{n_2} \in 
X$,
etc., such that: 

i) $\Psi\left( \pi_\Sigma\, \left(x^{[j]}_{i}\right) \right)= \vec{p}^j$, for every $1\le j \le q$, $1\le i \le n_j$;

ii) $x^{[1]}_1@ \dots @ x^{[1]}_{n_2} @ x^{[2]}_1@ \dots @ x^{[2]}_{n_2} @ \dots @x^{[q]}_{1}@ \dots @ x^{[q]}_{n_q} = u$. 
\end{proof}

\begin{lemma}
The consensual language $\mathcal{C}(D)$ is commutative. 
\end{lemma}
\begin{proof}
We notice first that $Y\lhd W$ and $X$ obviously verify the following two conditions: 
\begin{enumerate}[I)]
\item\label{en-shuffle} 
$Y\lhd W= \pi_{\widetilde a_1} (Y\lhd W) \shuffle \pi_{\widetilde a_2} (Y\lhd W) \shuffle \dots \shuffle 
\pi_{\widetilde a_k} (Y\lhd W)$; 
\item\label{en-shuffleX} 
if $x \in X$ then $\pi_{\widetilde a_1} (x) \shuffle \pi_{\widetilde a_2} (x) \shuffle \dots \shuffle 
\pi_{\widetilde a_k} (x) \subseteq X$.
\end{enumerate}
Let $u\in \mathcal{C}(D)$ and let $v\in \Sigma^+$ be such that 
$\Psi(v) = \Psi(u)$. 
Word $u$ is defined as $z@x_1 @ \dots @ x_n$, for some  $z \in Y\lhd W$, $n>0$ and some 
$x_1, \dots, x_n\in  X$. 
Word $v$ is a permutation of $u$, hence for all $a_i\in \Sigma$ $\pi_{a_i}(u) = \pi_{a_i}(v)$. 
By Prop.~(I) above, there exists a permutation $z'$ of $z$, such that 
$z'\in \scaffolding{R}_m\lhd W$, 
with $undot(z') = v$. 
Similarly, by Prop.~(II) above, for all $1\le j \le n$, there exists a permutation $x'_j$ of $x_j$ such that, 
for all $a_i\in\Sigma$, $\pi_{\widetilde a_i}(x'_j)= \pi_{\widetilde a_i}(x_j)$ and, moreover, such that 
$z' @ x'_i$ is  defined, with  $\pi_{\widetilde a_i}(z'@x'_i)=\pi_{\widetilde a_i}(z@x_i).$ Hence, also 
$z'@x'_1 @ \dots @ x'_n$ is defined, therefore $z'@x'_1 @ \dots @ x'_n= undot(z') = v$.  
\end{proof}
Next, Th.~\ref{th-Consens(D)=LIP} shows that $D$ consensually defines $L$, with $m$ and $r$ 
arbitrarily large. 

\begin{theorem}\label{th-Consens(D)=LIP}
For all even integers $m\ge 4$ and for every $R$ of the form $\{r\}$, with $1\le r \le m/2-1$, 
there exists 
 a decomposed base $D$ as in Def.~\ref{def-decomposedCOMLIP} such that 
the COM-LIP language $
L= 
\mathcal{C}\left(D \right).
$
\end{theorem}
\begin{proof}
Let $m,R, D,X,Y,W$ be
defined as in Def.~\ref{def-decomposedCOMLIP}, with $k = |\Sigma|, q = |\mathcal{P}|$. 
We first notice that, by definition of $Y \lhd W$ and of $X$:
\par\noindent
(*)\; if $z' \in Y$ then, for every $a_i \in \Sigma$, $ |z'|_{\widetilde a_i} $ is a multiple of $m$, 
$|z'|_{\mathring a_i}= 2\cdot |z|_{\widetilde a_i}/m$ and 
$|z'|_{a_i}= (m-2)\cdot  |z|_{\widetilde a_i}/m$.
\par\noindent
\textit{Proof of} 
$\mathcal{C}\left( D \right) \subseteq L $. 
Let $u \in  \mathcal{C}\left( D \right)$.  We show that $\Psi(u) \in \Psi(L)$.
Since $D$ is in decomposed form, $u$ must be the match of a word $z \in (Y\lhd W$) with $h\ge 0$ words 
 $x_1, \dots, x_h \in X$. Let $x = x_1@x_2 @\dots @ x_h$. Word $z$ has the form 
$z' \lhd w$ for some $z'\in Y$ and some $w \in W\subseteq \Sigma^*$. 
By Lm.~\ref{lm-LinearCombinOfPeriods}, Part (1), there exist $d_1, \dots , d_q \ge 0$ such that $\Psi(\pi_{\Sigma}(x)) = \vec{c}+ d_1 \cdot\vec{p}^{(1)}+\ldots +  d_q \cdot \vec{p}^{(q)}$.
 Also, by definition of $W$, there exist $q$ integers $0\le h_1, \dots , h_q <m/2$ such that 
$\Psi(w) = \vec{c} + h_1 \cdot\vec{p}^{(1)}\ldots +  h_q \cdot
\vec{p}^{(q)}$. 
Since $u = (z'\lhd w)@x$ is a strong match, $\Psi(u) = \Psi(\pi_\Sigma(z')) + \Psi(\pi_{\Sigma}(x)) + \Psi(\pi_{\Sigma}(w))$.
Notice that each component of  $\Psi(\pi_{\Sigma}(x))$ must be  even: 
by  $(z'\lhd w)@x$ being a strong match it follows that $|x|_{a_i}$ is equal to $|z'|_{\mathring a_i}$, which 
is even.
Again because $(z'\lhd w)@x$ is a strong match, $\Psi(\pi_\Sigma(z')) = (m-2)/2 \cdot \Psi(\pi_{\Sigma}(x))$.
Therefore: 
\[
\begin{array}{ll}
\Psi(u) &=  (m-2) \cdot \Psi(\pi_{\Sigma}(x))/2  + \Psi(\pi_{\Sigma}(x)) + \Psi(w) = \\
&= m\cdot \Psi(\pi_{\Sigma}(x)) + \Psi(w) = 
\\ &= m\cdot (d_1 \cdot\vec{p}^{(1)}+\ldots +  d_q \cdot
\vec{p}^{(q)}) + \vec{c} +  h_1 \cdot\vec{p}^{(1)}+\ldots +  h_q \cdot
\vec{p}^{(q)} = \\
&= \vec{c} + (m \cdot d_1 + h_1)\cdot\vec{p}^{(1)}+\ldots +  (m \cdot d_q + h_q)\cdot\vec{p}^{(q)} 
\end{array}
\]

Hence, $\Psi(u)\in\Psi(L)$.
\par\noindent
\textit{Proof of }
$L \subseteq 
\mathcal{C}\left( D \right)$. 
For all $u \in L$ 
there exist $q$ integers 
$n_1, \dots, n_q$  such that 
$\Psi(u)=
\vec{c} +n_1 \cdot\vec{p}^{(1)}+\ldots +  n_q \cdot
\vec{p}^{(q)}$.
For every $j$, $1\le j \le q$, let $h_j = n_j \text{ mod } (m/2)$. 
Let $d_j = n_j-h_j$ if $n_j p^{(j)}_i >0$, and $d_j = 0$ otherwise. Then, every $d_j$ and $h_j$ are such that $0\le h_j <m/2$ and $d_j$ is a (possibly zero) multiple of $m/2$.
By definition of $W$, there exists $w\in W$ such that $\Psi(w) = h_1 \cdot\vec{p}^{(1)}+\ldots +  h_q \cdot
\vec{p}^{(q)}$.
For all $a_i\in \Sigma$, let $z_i$ be the word in $(R_m(a_i))^*$ such that 
$|z_i| = d_1 p^{(1)}_i + \dots + d_q p^{(q)}_i$.
Such a word does exist, since each $d_j$ is a (possibly zero)  multiple of $m/2$, 
hence $d_1 p^{(1)}_i + \dots + d_q p^{(q)}_i$
is a multiple of $m/2$; 
if this multiple is 0, then $z_i=\epsilon$. By
definition of $R_m(a_i)$, word $z_i$ (when not empty) has, in every segment of length
$m$ belonging to $R_m(a_i)$,  exactly two occurrences of $\mathring a_i$, and $(m-2)$ occurrences 
of $a_i$. Hence, $|z_i|_{\mathring a_i} = 2(d_1 p^{(1)}_i + \dots + d_q p^{(q)}_i)/m$ and 
$|z_i|_{a_i} = (m-2)\cdot (d_1 p^{(1)}_i + \dots + d_q p^{(q)}_i)/m$. 
 \noindent We claim that there exists $z'\in Y$ such that 
$\Psi(undot(z')) =   d_1 \cdot\vec{p}^{(1)}+\ldots +  d_q \cdot
\vec{p}^{(q)}$. In fact, by Prop.~(*) above, there exists $z'\in Y$ such that 
$\pi_{\widetilde a_i}(z') = z_i$. Hence, $\Psi(\pi_\Sigma(z'))  = (m-2)\cdot (d_1 \cdot\vec{p}^{(1)}+\ldots +  d_q \cdot
\vec{p}^{(q)}$. 
By definition of $W$, there exists $w\in W$ such that 
\[\Psi(w) = \vec{c} + h_1 \cdot\vec{p}^{(1)}+\ldots +  h_q \cdot
\vec{p}^{(q)}.\] 
Let $z'' = \sw(z')$.  
By Lm.~\ref{lm-LinearCombinOfPeriods}, Part (2), there exist $n = 2d_1/m + 2d_2/m + \dots + 2d_q/m$ words $x_1, \dots, x_n \in 
X $ such that 
\[z''= x_1 @ \dots @ x_n, \text{ with } \Psi(\pi_\Sigma(z'')) =   
2\cdot (d_1 \cdot\vec{p}^{(1)}+\ldots +  d_q \cdot
\vec{p}^{(q)})/m.\] 
Consider now $x_i \lhd dot(w)$. This word is in $X$, since the fills included in $X$ may end with arbitrarily many $\mathring a$, for every $a \in \Sigma$. 
Clearly, from $x_i \lhd dot(w)$
one can obtain a strong match  $v$ with $z'\lhd w$: 
\begin{align*}
& v = (z'\lhd w) @ (x_1 \lhd dot(w)) @ \dots @ (x_n \lhd dot(w))\\
\text{with } &\Psi(v) =
\Psi(\pi_\Sigma(z')) + \Psi(\pi_\Sigma(z'')) + \Psi(\pi_\Sigma(w)) = 
 \Psi(u).
\end{align*}
 Since the language $\mathcal{C}(D)$ 
is commutative, and $v \in \mathcal{C}(D)$, also $u \in  \mathcal{C}(D)$.
\end{proof}
\par
\noindent We can now complete the proof of Th.~\ref{th-main}. 
Since a COM-SLIP language is the finite union of COM-LIP languages,  
 a \C language 
is the union and concatenation of COM-LIP languages. It can be assumed that these COM-LIP languages comply with Def.~\ref{def-decomposedCOMLIP} having only 
even components in every vector of the set $\mathcal{P}$ of periods (since otherwise they can be represented as the finite union of COM-LIP languages with this property). 
Select the same module and disjoint sets of slots   for the decomposed  bases of these COM-LIP languages. 
By Th.~\ref{th-join-concat}, since each COM-LIP is defined by a shiftable base with disjoint sets of slots, the various bases can be combined with $\cup$ and $\odot$, resulting in a shiftable base.  
By Th.~\ref{th-UnionDecomposForms} and
and Th.~\ref{th-ConcatDecomposForms}, 
the result is still a consensual language (with a decomposed base).
The inclusion is strict, since language $\{b a^1 b a^2 b a^3 \ldots  b a^{k}\mid k\geq 1\}$ has
a
non-SLIP commutative image, but it is in CREG~\cite{journals/ita/Crespi-ReghizziP11}.

\section{Related Work and Conclusion}\label{sectionRelWorkConcl}
By classical results, \C is included in the
class of languages recognized by  \emph{reversal-bounded} multi-counter machines~\cite{Baker1974,DBLP:journals/jacm/Ibarra78} (which is also closed under concatenation). 
The latter class  admits different, but  equivalent, 
 characterizations: as the class of languages recognized by 
(nondeterministic) \emph{blind MCMs'}~\cite{DBLP:journals/tcs/Greibach78a}, or as the minimal, intersection-closed full semi-AFL 
 including language $com((ab)^*)$~\cite{Baker1974,Greibach:1976:RCN}. However, the cited papers are not concerned with actual construction methods for the MCMs'.  
\par
 Although COM-SLIP languages have been much  
studied,  we are not aware of any specific study on
the effect on  COM-SLIP of  operations such as concatenation.
\par
Concerning the techniques to specify COM-SLIP   languages,  
 our   specification, using as patterns the commutative  Parikh vectors, bears some similarity to 
Kari's~\cite{Mateescu94}
``scattered deletion'' operation.
\par 
It is known that family COM-SLIP, when restricted  to a binary alphabet, is context-free~\cite{Latteux79,Rigo:2003}, therefore it enjoys
 closure under concatenation and star.  On the other hand, we observe that the intersection  
$I= L'^4\, \cap\, a^+ L'^2 b^+$, where $L'=com\left((ab)^+\right)$, is not context-free, since
\begin{center}
  $I\cap \left(a^+b^+\right)^4=\left\{a^n b^n  a^n b^n a^n b^n a^n b^n \mid n >1 \right\}$.
\end{center}
In~\cite{Rigo:2003},
the context-free grammar rules  for COM-LIP  again resemble our consensual specification. 
\par
Also, the   context-sensitive grammars in~\cite{jalc/Nagy09}, obtained by adding \emph{permutative} \emph{rules} of the form $AB \to BA$  to context-free grammars, 
include COM-SLIP and of course its closure by concatenation and star, but not its intersection with regular languages.
\par
Last, the COM-SLIP languages are included in the SLIP language family recognized by a formal device, based on so called restarting automata, studied in~\cite{jcss/NagyOttoO12}, but the grounds covered by   CREG  and by that family are quite different.
 Beyond the mentioned similarities,  we are unaware of anything related  to our congruence-based  decomposed form.
\paragraph{Unanswered questions}
This paper has added  a piece to our knowledge of the languages included in CREG; it has introduced a novel  compositional  construction for the union/concatenation, which is very general and hence likely to be useful  for other language subfamilies included in CREG.
 Some natural questions concern the closures of COM-SLIP under  other basic operations:  
is  the intersection of two COM-SLIP languages, or the Kleene star of a COM-SLIP language, in CREG?  
\par
A different kind of problem  is whether the only commutative languages that are in CREG are
semilinear; for instance,  the  nonsemilinear non-commutative  language $\{b a^1 b a^2 b a^3 \ldots  b a^{k}\mid k\geq 1\}$ is in CREG, but, for its commutative closure,  we do not know of a consensually regular specification.
Last, a more general problem is whether CREG is closed under union, concatenation, and star. A possible approach is to investigate whether every CREG language may be defined by a base which is joinable and shiftable, thus obtaining closure under union and concatenation by virtue of the lemmas presented in this paper. 
\bibliographystyle{eptcs}
\bibliography{automatabib}
\end{document}